\documentclass{easychair}
\usepackage[T1]{fontenc}
\usepackage[utf8]{inputenc}
\usepackage{graphicx}
\usepackage{wrapfig}
\usepackage{caption}
\usepackage{subcaption}
\usepackage{geometry}
\usepackage{amsmath,amsfonts,amssymb}
\usepackage{changepage}
\usepackage{eurosym}
\usepackage{dsfont}
\usepackage{stmaryrd}
\usepackage[ruled,linesnumbered]{algorithm2e}
\usepackage{turnstile}
\usepackage{mathrsfs}
\usepackage{enumitem}
\usepackage{tikz}
\usepackage{amsthm}
\usepackage[super]{nth}
\usepackage[toc,page]{appendix}
\usepackage{listings}
\usepackage{mathtools}
\usepackage{csquotes}
\usepackage{multicol}
\usepackage{fancyvrb}
\usepackage{makecell}
\usepackage[toc,page]{appendix}
\usepackage{mdframed}
\usepackage{setspace}
\usepackage{tcolorbox}
\usepackage{authblk}

\usepackage[
backend=biber,
style=numeric,
sortlocale=fr_FR,
natbib=true,
]{biblatex}

\usetikzlibrary{automata, positioning, calc, shapes}
\tikzset{%
	do path picture/.style={%
		path picture={%
			\pgfpointdiff{\pgfpointanchor{path picture bounding box}{south west}}%
			{\pgfpointanchor{path picture bounding box}{north east}}%
			\pgfgetlastxy\x\y%
			\tikzset{x=\x/2,y=\y/2}%
			#1
		}
	},
	plus/.style={do path picture={    
			\draw [line cap=round] (-3/4,0) -- (3/4,0) (0,-3/4) -- (0,3/4);
	}}
}

\allowdisplaybreaks

\usetikzlibrary{automata, positioning, calc, shapes}

\theoremstyle{plain}
\newtheorem{fact}{Fact}
\newtheorem{remark}{Remark}
\newtheorem{theorem}{Theorem}
\newtheorem{lemma}{Lemma}

\newtheorem{example}{Example}

\newtheorem{proposition}{Proposition}
\newtheorem{definition}{Definition}

\renewcommand{\phi}{\varphi}

\DeclareMathOperator{\binder}{\downarrow}

\newcommand{\poison}{\bullet}

\newcommand{\poisonRelation}{\stackrel{\poison}{\rightarrow}}
\newcommand{\equivModel}{\stackrel{\poisonAtom}{\leftrightsquigarrow}}
\newcommand{\bisimulation}{\stackrel{\poisonAtom}{\rightleftharpoons}}

\newcommand{\playerName}{\mathbb{P}}
\newcommand{\deamonName}{\mathbb{O}}
\newcommand{\poisonAtom}{\mathfrak{p}}

\newcommand{\poisonAtomPred}{\mathfrak{P}}
\newcommand{\propSet}{{\bf P}}
\newcommand{\nomSet}{{\bf N}}

\newcommand{\Model}{\mathcal{M}}
\newcommand{\modelPML}{\mathcal{M}}

\newcommand{\languagePML}{\mathcal{L}^\poisonAtom}

\newcommand{\poisonTheory}{\mathbb{T}^\poisonAtom}

\newcommand{\FOL}{\ensuremath{\mathsf{FOL}}}
\newcommand{\SML}{\ensuremath{\mathsf{SML}}}
\newcommand{\PML}{\ensuremath{\mathsf{PML}}}

\newcommand{\K}{\ensuremath{\mathsf{K}}}

\newcommand{\lanFOL}{\mathcal{L}}
\newcommand{\languageHL}{\mathcal{L}^{\mathcal{H}(\binder)}}

\newcommand{\HLSatisfaction}{\models_{{\bf H}}}

\newcommand{\ar}{\ensuremath{\shortrightarrow}}
\newcommand{\tuple}[1]{( #1 )}

\newcommand{\IFF}{\Longleftrightarrow}

\newcommand{\bu}{\lbox{\univ}} %

\newcommand{\dl}{\lozenge} %
\newcommand{\bl}{\square} %

\newcommand{\limp}{\rightarrow}
\newcommand{\lbox}[1]{[#1]}

\newcommand{\univ}{\ensuremath{U}}

\newcommand{\PRO}{\playerName}
\newcommand{\OPP}{\deamonName}

\title{Credulous Acceptability, Poison Games and Modal Logic}

\author{
	Davide Grossi\inst{1}
	\and
	Simon Rey\inst{2}$^*$
}

\institute{
	University of Groningen,
	Groningen, Netherland\\
	\email{d.grossi@rug.nl}
	\and
	ENS Paris-Saclay,
	Cachan, France\\
	\email{srey@ens-paris-saclay.fr}
}

\authorrunning{Grossi, Rey}

\titlerunning{Credulous Acceptability, Poison Games and Modal Logic}

\begin{document}	
	\maketitle

	\begin{abstract}
		The Poison Game is a two-player game played on a graph in which one player can influence which edges the other player is able to traverse. It operationalizes the notion of existence of credulously admissible sets in an argumentation framework or, in graph-theoretic terminology, the existence of non-trivial semi-kernels. We develop a modal logic (poison modal logic, \PML) tailored to represent winning positions in such a game, thereby identifying the precise modal reasoning that underlies the notion of credulous admissibility in argumentation. We study model-theoretic and decidability properties of \PML, and position it with respect to recently studied logics at the cross-road of modal logic, argumentation, and graph games. 
	\end{abstract}

	\section{Introduction}
	
	In abstract argumentation theory \cite{baroni09semantics,baroni11introduction}, an argumentation framework (or {\em attack graph}) \cite{dung95acceptability} is a directed graph $\tuple{A, \ar}$  where $A$ is a set of nodes (or {\em arguments}) and $\ar \subseteq A^2$ is a set of directed edges (or {\em attacks}). For $x, y \in A$ such that $x \ar y$ we say that $x$ attacks $y$. An {\em admissible} set \cite{dung95acceptability}, of a given attack graph, is a set $X \subseteq A$ such that: (a) no two nodes in $X$ attack one another; and (b) for each node $y \in A \backslash X$ attacking a node in $X$, there exists a node $z \in X$ attacking $y$. That is, $X$ is internally coherent, and can counterattack any attack moved to any of its arguments. Such sets are also called {\em semi-kernels} in the theory of directed graphs \cite{galeana84kernels}. More precisely, if $X$ is an admissible set of $\tuple{A, \ar}$, then it is a semi-kernel of the directed graph obtained by inverting the attack relation $\ar$ (i.e., the `being attacked' graph), and vice versa. These sets form the basis of all main argumentation semantics first developed in  \cite{dung95acceptability} and they are central to the influential graph-theoretic systematization of logic programming and default reasoning pursued in \cite{dimopoulos94graph, dimopoulos96graph}, where they have been proven to correspond to the so-called partial stable models of logic programming \cite{przymusinzki90extended}. 
	
	\vspace{-0.1cm}
	\paragraph{Contribution}
	Given the importance of admissible sets in argumentation theory, one of the key reasoning tasks in abstract argumentation consists in deciding whether any given argumentation framework contains non-empty admissible sets. In the terminology of argumentation, this amounts to deciding whether the framework contains any {\em credulously admissible} arguments. The property corresponds in turn to the existence of non-trivial semi-kernels in the inverted attack graph. Credulous acceptability is a benchmark semantics for the evaluation of arguments in abstract argumentation \cite{bench-capon07argumentation}. Interestingly, the notion has an elegant operationalization in the form of two-player games, called {\em Poison Game} in the graph theory literature \cite{Duchet_Meyniel_1993}, and {\em game for credulous acceptance} in the argumentation theory literature \cite{vreeswijk00credulous,modgil09proof}. The poison game is the starting point of the paper. Inspired by it we define a new modal logic, called {\em poison modal logic} (\PML), whose operators capture the strategic abilities of players in the Poison Game, and are therefore fit to capture the modal reasoning involved in the notion of credulous admissibility. This answers, at least in part, a research question left open in \cite{grossi10logic}. The paper studies \PML{} by: defining a suitable notion of bisimulation for it, which in turn answers another open question \cite{gabbay14when} concerning the logic of credulous admissibility, namely a notion of structural equivalence tailored for it; establishing a first-order characterization result for \PML{}  in the tradition of \cite{VanBenthem_1983}; proving the undecidability of satisfiability in a multi-modal variant of \PML; and exploring its links with hybrid \cite{Blackburn_2006_handbook} and memory logics \cite{areces08expressive,areces11expressive}. More broadly we see the paper as a contribution to bridging, in a systematic way, concepts from abstract argumentation theory \cite{dung95acceptability}, games on graphs \cite{berge96combinatorial} and modal logic \cite{Blackburn_2001}.
	
	\vspace{-0.1cm}
	\paragraph{Related work}
	The paper is a natural continuation of the line of work interfacing abstract argumentation and modal logic \cite{grossi10logic,grossi11argumentation,gabbay14when,grossi14justified,shi17argument,shi18beliefs}, which focuses on the modal logic characterization of key argumentation-theoretic notions, and their analysis through model and proof-theoretic tools. A first bimodal dynamic logic tailored to the poison game was introduced in \cite{Mierzewski_Blando_2016}, where two modalities are used to keep track of which parts of the underlying graph are accessible to each player. Our approach is somewhat simpler and based on the combination of one classical and one dynamic modality. \PML{} is also directly related to so-called memory logics, extensively studied in the last decade \cite{areces08expressive,areces11expressive}. In fact \PML{} can be thought of as a modal logic with two operators: a standard one, and one which 'memorizes' the states which are reached by traversing an edge of the underlying frame. 
	The paper relates also to the research program investigating the modal logic theory of graph games, sparked by the recent work on sabotage modal logic (\SML{}, \cite{VanBenthem_2005, loeding03model,rohde06mu,areces13tableaux,areces15relation, Aucher_al_2017}). \SML{} was tailored to capture the logic behind winning strategies in a specific two-player, perfect-information, zero-sum game played on graphs, known as the sabotage game \cite{VanBenthem_2005}.  Like \SML{}, \PML{} sits at the intersection of two well-established lines of research in modal logic: dynamic epistemic logic \cite{VanDitmarsch_2007} and the logical dynamics tradition it generated, which is broadly concerned with the study of operators interpreted on transformations of semantics structures \cite{Aucher_al_2009_global,Kooi_Renne_2011,benthem11logical,areces15relation}; and game logics concerning the logical analysis of games \cite{VanBenthem_2014,benthem10game}. 
	
	\vspace{-0.1cm}
	\paragraph{Outline}
	The article is organized as follow. First we set the ground in Section \ref{sec:Preliminaries} by showing how the standard modal language can already capture the key logic behind statements of this type: ``set $X$ is a semi-kernel'' (of a given directed graph), and by introducing the Poison Game. Section \ref{sec:PML} introduces \PML{} and establishes some basic facts.
	Section \ref{sec:FOL} gives a translation into First Order Logic (FOL) which is invariant for the poison bisimulation as defined in Section \ref{sec:PoisonBisi}. Decidability is addressed in Section \ref{sec:Decidability}. Section \ref{sec:conclusions} explores links between \PML{} and other logical frameworks and concludes.

	\section{Preliminaries} \label{sec:Preliminaries}
	
	We start by providing some preliminaries on existing bridges between modal logic and abstract argumentation, and a concise presentation of the Poison Game.
	
	\subsection{Modal Logic and Credulous Admissibility}
	
	As hinted in the introduction, and following \cite{grossi10logic} we study attack graphs $\tuple{A, \ar}$ through their inversions (the `being attacked' graphs) $\tuple{A, \ar^{-1}}$ which we view as Kripke frames \cite{Blackburn_2001} $\tuple{W, R}$ where, $W = A$, and $R = \ar^{-1}$. So, writing $w R w'$ stands for argument $w$ is attacked by argument $w'$.\footnote{In what follows we will refer to $\tuple{W, R}$ also as attack graphs even though, technically speaking, they are inversions of attack graphs.} In this view, a Kripke model $M = \tuple{W, R, V}$, where $V$ is a valuation function $V: \propSet \to 2^W$, can be thought of an argumentation framework where propositional labels in $\propSet$ are assigned to sets of arguments. The standard modal language,
	\[
	\mathcal{L} :\phi ::= p \mid \neg \phi \mid \phi \wedge \phi \mid \lozenge \phi,
	\]
	becomes therefore a language in which it is possible to express properties of argumentation frameworks. Through the standard modal semantic clause
	\begin{align}
	(\tuple{W,R,V}, w) \models \dl \phi \Leftrightarrow \exists w' \in W, w R w', (\tuple{W,R,V}, w') \models \phi,
	\end{align}
	formulas $\lozenge \phi$ are statements of the type ``the current argument is attacked by an argument in the set of arguments denoted by $\phi$''. Specifically, as shown in \cite{grossi10logic} a number of key argumentation-theoretic properties are expressible in the standard modal language extended with the universal modality $\bu$ (known as logic $\K^U$ \cite{Blackburn_2001}). In particular, formula 
	\begin{align}
	\bu (p \limp \neg \dl p) \land \bu (p \limp \bl\dl p)
	\end{align}
	expresses the property ``the set denoted by $p$ under function $V$ is admissible'' (in the underlying argumentation framework) \cite{grossi10logic}. So a well-studied logic such as $\K^U$ suffices to express that a given set of arguments is admissible.
	However, the existence of credulously admissible arguments has an obvious second-order flavor and the question of whether it could be modally expressed without resorting to second-order modal logic remained an open question in \cite{grossi10logic}.

	\subsection{The Poison Game}
	
	The Poison Game was introduced in \cite{Duchet_Meyniel_1993} to characterize the existence of non-empty semi-kernels in directed graphs. A very similar game was later independently introduced in \cite{vreeswijk00credulous} to characterize credulous admissibility of arguments in argumentation frameworks.\footnote{A detailed comparison of the two games is not in the scope of this paper. Albeit very similar, the two games are from a technical point of view slightly different, and are adequate with respect to slightly different notions: the poison game is adequate w.r.t. the existence of non-empty admissible sets; the game from \cite{vreeswijk00credulous} is adequate w.r.t. the membership of one given argument to at least one admissible set.} Our presentation of the game follows that of \cite{Duchet_Meyniel_1993}. 
	
	The Poison Game is a two-player ($\PRO$, the proponent, and $\OPP$, the opponent), win-lose, perfect-information game \cite{osborne94course} played on a directed graph $\tuple{W,R}$. The game starts by $\PRO$ selecting a node $w \in W$. After this initial choice, $\OPP$ selects a successor of the node picked by $\PRO$, $\PRO$ then selects a successor of the node picked by $\OPP$ and so on. However, while $\OPP$ can choose any successor of the current node, $\PRO$ can only select successors which have not yet been visited---{\em poisoned}---by $\OPP$. $\OPP$ wins if and only if $\PRO$ ends up in a position with no available successors. In all other cases the game is won by $\PRO$. 

	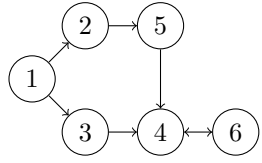
\begin{figure}[t]
		\centering
		\begin{tikzpicture}
		\node[draw,circle]	(a)  {$1$};
		\node[draw,circle]	(b)  [above right of=a]  {$2$};
		\node[draw,circle]	(c)  [below right of=a] {$3$};
		\node[draw,circle]	(e)  [right of=c] {$4$};
		\node[draw,circle]	(d)  [right of=b] {$5$};
		\node[draw,circle]	(f)  [right of=e] {$6$};
		
		\path[->] (a) edge	node {} (b);
		\path[->] (a) edge	node {} (c);
		\path[->] (b) edge	node {} (d);
		\path[->] (c) edge	node {} (e);
		\path[<->] (e) edge	node {} (f);
		\path[->] (d) edge	node {} (e);
		
		\end{tikzpicture}
		\caption{Graph of Examples \ref{example:game1} and \ref{example:game2}. The graph is (the inversion of) a framework discussed in \cite{vreeswijk00credulous}.}
		\label{figure:backtrack}
	\end{figure}
	
	\begin{example} 
		\label{example:game1}
		A possible run of the Poison Game on the graph depicted in Figure \ref{figure:backtrack} is: $\PRO$ starts by selecting node 1, then $\OPP$ moves to and poisons node 3, $\PRO$ answers by moving to 4, $\OPP$ in returns moves to 6 and from there the game will repeat the two last moves indefinitely. Player $\PRO$ therefore wins the game.
	\end{example}
	
	What makes this game interesting is that the existence of a winning strategy for $\PRO$, if $\tuple{W,R}$ is finite\footnote{The result holds also with a weaker condition requiring every weakly connected component of the directed graph to be finite.}, is equivalent to the existence of a (non-empty) semi-kernel in the graph or, in the argumentation terminology, the existence of credulously admissible arguments in the inverted graph $\tuple{W,R^{-1}}$.
	
	\begin{theorem}[\cite{Duchet_Meyniel_1993} \cite{Duchet_Meyniel_1993}]
		\label{thm:DuchetMeyniel}
		Let $\tuple{W,R}$ be a finite directed graph. There exists a non-empty semi-kernel in $\tuple{W,R}$ if and only if $\PRO$ has a winning strategy in the Poison Game for $\tuple{W,R}$.
	\end{theorem}
	\begin{proof}[Sketch of proof]
		\fbox{Left-to-right} If a non-empty semi-kernel $X \subset W$ exists, then $\PRO$ can win the game simply by picking the initial node in $X$ and then responding to each move of $\OPP$ with a successor in $X$, which is guaranteed to exist since $X$ is a semi-kernel.
		\fbox{Right-to-left} If $\PRO$ has a winning strategy, she can play indefinitely no matter what $\OPP$ does. As $W$ is finite, this means that $\PRO$ visits a finite set of states infinitely often. Call such set $X$. It suffices to show that $X$ is indeed a semi-kernel. Clearly no state in $X$ has a successor in $X$ (i.e., $X$ is independent), as otherwise such successor would have been poisoned by $\OPP$. Moreover for each state $x \in X$, for any successor $y \in W \backslash X$ of $x$ that can be selected by $\OPP$, there exists a successor $z$ of $y$ that can be selected by $\PRO$ infinitely often, therefore belonging to $X$. 
	\end{proof}
	
	\begin{example}
		\label{example:game2}
		In the graph of Figure \ref{figure:backtrack}: $\{4\}$ and $\{6\}$ are two semi-kernels. $\PRO$ has several winning strategies in the game played on this graph. She can choose $1$ and force the infinite run described in Example \ref{example:game1}. Alternatively she could simply choose $4$ or $6$ and again force an infinite run.
	\end{example}

	\section{Poison Modal Logic (\PML)} \label{sec:PML}
	
	This section introduces the syntax and semantics of \PML, discusses some of its validities and some properties it is able to express. The language we propose is directly motivated by the Poison Game: the standard modality $\lozenge$ tracks moves of $\PRO$ selecting successors of a current state, the novel poison modality $\blacklozenge$ tracks moves of $\OPP$ selecting successors of a current state {\em and} poisoning them.

	\subsection{Syntax \& Semantics}
	
	The poison modal language $\languagePML$ is defined by the following BNF:
	\[
	\languagePML: \phi ::= p \mid \neg \phi \mid (\phi \wedge \phi) \mid \lozenge \phi \mid \blacklozenge \phi,
	\]
	where $p \in \propSet \cup \{\poisonAtom{}\}$ with $\propSet$ a countable set of propositional atoms and $\poisonAtom$ a distinguished atom called {\em poison atom}. We will also discuss multi-modal variants of the above language, denoted $\languagePML_n$, where $n \geq 1$ denotes the number of distinct pairs $(\lozenge_i,\blacklozenge_i)$ of modalities, with $1 \leq i \leq n$ and where each $\blacklozenge_i$ comes equipped with a distinct poison atom $\poisonAtom_i$.
	
	This language is interpreted on Kripke models $\modelPML = (W, R, V)$, as defined above. When confusion may arise, we will write $W^\modelPML$, $R^\modelPML$ and $V^\modelPML$ to refer to the elements of the model $\modelPML$. We will note $w \in \modelPML$ to say $w \in W^\modelPML$. A pointed model is a pair $(\modelPML, w)$ with $w\in \modelPML$. We call $\mathfrak{M}$ the set of all pointed models and $\mathfrak{M}^\emptyset$ the set of pointed models $(\modelPML, w)$ such that $V^\modelPML(\poisonAtom) = \emptyset$, that is the class of pointed models where no state satisfies $\poisonAtom$.

	We define now an operation $\poison$ on models which, given an input model and a state, modifies its function $V$ by adding that state to $V(\poisonAtom)$. Formally, for $\modelPML = (W, R, V)$ and $w \in W$:
	\[
	\modelPML_w^\poison = (W, R, V)_w^\poison = (W, R, V'),
	\]
	where $\forall p \in \propSet, V'(p) = V(p)$ and $V'(\poisonAtom) = V(\poisonAtom) \cup \{w\}$. We are now equipped to formally define the semantics of $\languagePML$.
	
	\begin{definition}[Satisfaction relation]
		\label{def:semanticsPML}
		Let $(\modelPML , w) \in \mathfrak{M}$.
		The satisfaction relation of \PML{} is defined recursively as follows:
		\begin{align*}
		(\modelPML, w) & \models p \Longleftrightarrow w \in V(p), \forall p \in \propSet \cup \{\poisonAtom\} \\
		(\modelPML, w) & \models \neg \phi \Longleftrightarrow (\modelPML, w) \not \models \phi \\
		(\modelPML, w) & \models \phi \wedge \psi \Longleftrightarrow (\modelPML, w) \models \phi \text{ and } (\modelPML, w) \models \psi \\
		(\modelPML, w) & \models \lozenge \phi  \Longleftrightarrow \exists v \in W, wRv, (\modelPML, v) \models \phi \\
		(\modelPML, w) & \models \blacklozenge \phi \Longleftrightarrow \exists v \in W, wRv, (\modelPML_v^\poison, v) \models \phi.
		\end{align*}
	\end{definition}
	The poison formula $\blacklozenge \phi$ is then true at $w$ in $\modelPML$ if and only if $\phi$ is true at a successor $w'$ of $w$ in the model obtained from $\modelPML$ by adding $w'$ to the valuation of the poison atom $\poisonAtom$. Validity in a model and in a frame are defined in the usual way, but the relevant class of models to specify \PML{} is $\mathfrak{M}^\emptyset$, that is, those models where $\poisonAtom$ starts with an empty valuation. \PML{} is therefore the set of formulas which are valid in $\mathfrak{M}^\emptyset$. Similarly, $\PML_n$ is the set of formulas of $\languagePML_n$ which are valid in $\mathfrak{M}^\emptyset$.
	We introduce some auxiliary definitions.
	\begin{definition}[Poison modal theory]
		\label{def:poisonModalTheory}
		The poison modal theory of a pointed model $(\modelPML, w) \in \mathfrak{M}$ is the set $\poisonTheory(\modelPML, w) \subseteq \languagePML$ of formulas defined as follows:
		$$
		\poisonTheory(\modelPML, w) = \{\phi \in \languagePML \mid (\modelPML, w) \models \phi\}.
		$$
	\end{definition}
	
	\begin{definition}[Poison relation] \label{def:poisonRelation}
		The poisoning relation $\poisonRelation$ between two pointed models is defined as:
		\begin{align*}
		(\modelPML, w) \poisonRelation (\modelPML', w') & \IFF wR^\modelPML w'\text{ and }\modelPML' = \modelPML^\poison_{w'}.
		\end{align*}
		Furthermore, we denote $(\modelPML, w)^\poison$ the set of all pointed models accessible from $\modelPML$ via a poisoning relation, formally: 
		\begin{align*}
		(\modelPML, w)^\poison & = \{(\modelPML', w') \mid(\modelPML, w) \poisonRelation (\modelPML', w')\}.
		\end{align*}
	\end{definition}
	
	\begin{definition}[Poison modal equivalence]
		Two pointed models $(\modelPML, w)$ and $(\modelPML', w')$ 	are poison modally equivalent---in symbols, $(\modelPML, w) \equivModel (\modelPML', w')$---if and only if, $\forall \phi \in \languagePML$:
		\[ 
		(\modelPML, w) \models \phi \IFF (\modelPML', w') \models \phi.
		\]
	\end{definition}

	\subsection{Validity and Expressivity: Examples}

	\begin{fact}
		Let $p\in \propSet$ and $\phi, \psi \in \languagePML$. The following formulas are validities of \PML{} (w.r.t. class $\mathfrak{M}^\emptyset$):
		\begin{align}
		\label{eq:validity0}
		& \neg \poisonAtom \land \blacksquare \poisonAtom \\
		\label{eq:validity8}
		&\square \bot \rightarrow \blacksquare \phi \\
		\label{eq:validity3}
		&\blacksquare p \leftrightarrow \square p \\
		\label{eq:validity4}
		&\square \poisonAtom \rightarrow (\blacksquare \phi \leftrightarrow \square \phi) \\
		\label{eq:validity6_wedgeAxiom}
		&\blacksquare (\phi \wedge \psi) \leftrightarrow (\blacksquare \phi \wedge \blacksquare \psi) \\
		\label{eq:validity7_negAxiom}
		&\blacksquare \neg \phi \rightarrow (\square \bot \vee \neg \blacksquare \phi)
		\end{align}
	\end{fact}
	Proofs are omitted. It can also be immediately noticed that \PML{} is not closed under uniform substitution. For instance, the schematic version $\blacksquare \phi \leftrightarrow \square \phi$ of Formula \eqref{eq:validity3} is clearly invalid.

	\medskip
	
	To illustrate the expressive power of \PML, we show that it is possible to express the existence of circuits \cite[page 4]{bollobas2012graph} in the modal frame, a property not expressible in the standard modal language. A circuit is a sequence of nodes such that two consecutive nodes are adjacent and the first and the last nodes are the same, compared to a cycle we allow here for any node to appear multiple times. Consider the class of formulas $\delta_n$, with $n \in \mathbb{N}_{>0}$, defined inductively as follows, with $i < n$: [Base] $\delta_1 = \lozenge\poisonAtom$; [Step] $\delta_{i+1} = \lozenge (\neg \poisonAtom \wedge \delta_i)$.
	
	\begin{fact}
		\label{prop:CycleDetection}
		Let $(\modelPML, w) \in \mathfrak{M}^\emptyset$ with $\modelPML = (W, R, V)$, then for $i, n \in \mathbb{N}_{>0}$  there exists $w\in W$ such that $(\modelPML, w) \models \blacklozenge \delta_n $ if and only if there exists a circuit \cite[page 4]{bollobas2012graph} of length $n$ in the frame $(W, R)$.
	\end{fact}
	\begin{proof}
		Observe that the formula $\blacklozenge \delta_n$ has only one occurrence of the poison modality $\blacklozenge$.
		As $V(\poisonAtom) = \emptyset$ by assumption, the only poisoned state when we go through the formula is the one poisoned by $\blacklozenge$. The formula then states that one can reach that unique poisoned state without passing through other poisoned states in $n$ steps. It follows that a cycle exists whose length $i$ is $n$ or smaller.
	\end{proof}
	
	A direct consequence of Fact \ref{prop:CycleDetection} is that \PML{} is not bisimulation invariant. In particular, its formulas are not preserved by tree-unravelings and it does not enjoy the tree model property.

	\subsection{Winning Strategies of the Poison Game}

	\PML{} can express winning positions (that is, states in a graph in which a player has a winning strategy) in a natural way. Given a frame $\tuple{W,R}$, nodes satisfying formulas $\blacklozenge \square \poisonAtom$ are winning for $\OPP$ as she can move to a dead end for $\PRO$. So are also nodes satisfying formula $\blacklozenge \square \blacklozenge \square \poisonAtom$: she can move to a node in which, no matter which successor $\PRO$ chooses, she can then push her to a dead end. In general, winning positions for $\OPP$ are defined by the following infinitary $\languagePML$-formula:
	\begin{align}
	\blacklozenge \square \poisonAtom \vee \blacklozenge \square \blacklozenge \square \poisonAtom \vee \ldots \label{eq:OPP}
	\end{align}
	Dually, winning positions for $\PRO$ are defined by the following infinitary $\languagePML$-formula:
	\begin{align}
	\blacksquare \lozenge \neg \poisonAtom \land \blacksquare \lozenge \blacksquare \lozenge \neg \poisonAtom \land \ldots \label{eq:PRO}
	\end{align}
	
	\begin{remark}[Credulous admissibility and \PML] \label{rem:credadm}
		By Theorem \ref{thm:DuchetMeyniel}, formula \eqref{eq:PRO}, interpreted on the inversion of an argumentation framework, expresses the property ``there exist credulously admissible arguments in the framework''. To the best of our knowledge, this is the first modal characterization of the notion, albeit an infinitary one.\footnote{Formulas $\eqref{eq:OPP}$ and $\eqref{eq:PRO}$ call naturally for a fixpoint extension of \PML. Such an extension poses interesting technical challenges very similar to those charted in \cite{Aucher_al_2017} for a $\mu$-calculus extension of sabotage modal logic.}
	\end{remark}

	\section{Expressivity of \PML} \label{sec:FOL}

	\subsection{Translation into First-Order Logic}
	
	Let $\lanFOL$ be the language of the binary fragment of first-order logic (\FOL) with equality. We present here a translation of the language of \PML{} into $\lanFOL$.
	\begin{definition}[\FOL{}translation]
		\label{def:FOLTrans}
		Let $p, q, \ldots \in \propSet$ be propositional atoms, we call $P, Q, \ldots$ their corresponding first-order predicate. The first-order predicate for the poison atom $\poisonAtom$ is $\poisonAtomPred$.
		Let $N$ be a finite set of variables, and $x$ a designated variable, the translation $ST^N_x : \languagePML \rightarrow \lanFOL$ is defined inductively as follows:
		\begin{align*}
		ST^N_x(p) & = P(x), \forall p \in \propSet \\
		ST^N_x(\neg \phi) & = \neg ST^N_x(\phi) \\
		ST^N_x(\phi \wedge \psi) & = ST^N_x(\phi) \wedge ST^N_x(\psi) \\
		ST^N_x(\lozenge \phi) & = \exists y\left(xRy \wedge ST^N_y(\phi)\right) \\
		ST^N_x(\blacklozenge \phi) & = \exists y\left(xRy \wedge ST^{N \cup \{y\}}_y(\phi)\right) \\
		ST^N_x(\poisonAtom) & = \poisonAtomPred(x) \vee \bigvee_{y \in N} (y = x).
		\end{align*}
	\end{definition}
	The definition is naturally extended to inputs consisting of sets of formulas. Let us briefly comment on the translation. A state is poisoned either if it is in the valuation of $\poisonAtom$, or if it has been poisoned by traversing a link instantiating the semantics of the poison modality, in which case the world is added to $N$ which `book-keeps' the set of poisoned states.
	It is worth noticing that the translation does not, in general, returns a formula with only one free variable. It does so, however, when $N$ is set to $\emptyset$. We move now to proving that the translation is correct.
	
	\begin{lemma}For a model $\modelPML$ and an assignment $g$:
		\label{lem:FOLTrans}
		$$\modelPML^\poison_w \models ST^N_x(\phi)[g] \Longleftrightarrow \modelPML \models ST^{N\cup \{y\}}_x(\phi)[g_{y:=w}].$$
	\end{lemma}
	\begin{proof}[Sketch of proof]
		We prove the lemma by induction on the structure of $\phi$ (standard cases are omitted).
		
		\fbox{$\phi = \poisonAtom$} This case (part of the induction base) is established by the following series of equivalences, using the definitions of the $^\poison$ operation on models and of the standard translation. 
		\begin{align*}
		\modelPML^\poison_w \models ST^N_x(\poisonAtom)[g] & \Leftrightarrow \modelPML^\poison_w \models \left(\poisonAtomPred(x) \vee \bigvee_{y \in N} (y = x)\right)[g] \\
		& \Leftrightarrow \modelPML \models \left(\poisonAtomPred(x) \vee \bigvee_{y \in N \cup \{w\}} (y = x)\right)[g] \\
		& \Leftrightarrow \modelPML \models ST_x^{N \cup \{y\}}(\poisonAtom)[g_{y := w}].
		\end{align*}
		
		\fbox{$\phi = \blacklozenge \psi$} with $\psi \in \languagePML$. The case is established by the following series of equivalences, using the definitions of the $^\poison$ operation on models, of the standard translation, the semantics of $\blacklozenge$ and $\land$, and the induction hypothesis. 
		\begin{align*}
		\modelPML^\poison_w \models & ST^N_x(\blacklozenge \psi)[g] 
		\Leftrightarrow \modelPML^\poison_w \models \exists y \left(xRy \wedge ST_y^{N\cup\{y\}}(\psi)\right)[g] \\ %
		& \Leftrightarrow \exists v, g(x)Rv,  \modelPML^\poison_w \models ST^{N \cup \{y\}}_y(\psi)[g_{y := v}] \\
		& \Leftrightarrow \exists v, g(x)Rv, \modelPML \models ST^{N \cup \{y, z\}}_y(\psi)[g_{y := v, z := w}] \\
		& \Leftrightarrow \modelPML \models \exists y \left(xRy \wedge ST_y^{N\cup\{y, z\}(\psi)}\right)[g_{z := w}] \\
		& \Leftrightarrow \modelPML \models ST^{N\cup\{z\}}_x(\blacklozenge \psi)[g_{z := w}]. \\
		\end{align*}
		This completes the proof.	
	\end{proof}

	\begin{theorem}
		\label{thm:correctTranslFOL}
		Let $(\modelPML, w)$ be a pointed model and $\phi \in \languagePML$ a formula, we have then:
		$$(\modelPML, w) \models \phi \Longleftrightarrow \modelPML \models ST^\emptyset_x(\phi)[x:=w].$$
	\end{theorem}
	\begin{proof}[Sketch of proof]
		The proof is by induction on the structure of $\phi$ (standard cases are omitted). 
		
		\fbox{$\phi = \poisonAtom$} This case (part of the base case) is established through the following series of simple equivalences:
		\begin{align*}
		(\modelPML, w) \models \poisonAtom & \Leftrightarrow \modelPML \models \poisonAtomPred(x)[x:=w] \\
		& \Leftrightarrow \modelPML \models \poisonAtomPred(x) \vee \bigvee_{y \in \emptyset} (y = x) [x:=w] \\
		& \Leftrightarrow \modelPML \models ST^\emptyset_x(\poisonAtom)[x := w].
		\end{align*}
		
		\fbox{$\phi = \blacklozenge \psi$} with $\psi \in \languagePML$. This case is established through the following series of equivalences, using the semantics of $\blacklozenge$, the definition of the standard translation and Lemma \ref{lem:FOLTrans}:
		\begin{align*}
		(\modelPML, w) \models \blacklozenge \psi & \Leftrightarrow \exists v, wRv, (\modelPML^\poison_v, w) \models \psi \\
		& \Leftrightarrow \exists v, wRv, \modelPML^\poison_v \models ST^\emptyset_x(\psi)[x := w] \\
		& \Leftrightarrow \exists v, wRv, \modelPML \models ST^{\{y\}}_x(\psi)[x := w, y := v] \\ %
		& \Leftrightarrow \modelPML \models \exists y \left(xRy \wedge ST^{\{y\}}_x(\psi)\right)[x := w] \\
		& \Leftrightarrow \modelPML \models ST^\emptyset_x(\blacklozenge \psi)[x := w].
		\end{align*}
		This completes the proof.
	\end{proof}

	\subsection{Poison Bisimulation} \label{sec:PoisonBisi}
	
	As observed earlier \PML{} is not bisimulation invariant. In what follows we define a notion of bisimulation tailored to \PML.
	
	\begin{definition}[p-bisimulation]
		\label{def:pbisimulation}
		Two pointed models $(\modelPML_1, w_1)$ and $(\modelPML_2, w_2)$ are said to be p-bisimilar, written $(\modelPML_1, w_1) \bisimulation (\modelPML_2, w_2)$, if there exists a relation $Z \subseteq W^{\modelPML_1} \times W^{\modelPML_2}$ (the p-bisimulation relation) such that $w_1 Z w_2$ and, for any states $w \in W^{\modelPML_1}$ and $v \in W^{\modelPML_2}$, whenever $wZv$ the following clauses are satisfied:
		\begin{description}
			\item[Atom:] For any atom $p \in \propSet\cup\{\poisonAtom\}$, $w \in V^{\modelPML_1}(p)$ iff $v \in V^{\modelPML_2}(p)$.
			
			\item[Zig$_\lozenge$:] If there exists $w' \in W^{\modelPML_1}$ such that $wR^{\modelPML_1}w'$ then there exists $v' \in W^{\modelPML_2}$ such that $vR^{\modelPML_2}v'$ and $(\modelPML_1, w')Z(\modelPML_2, v')$.
			
			\item[Zag$_\lozenge$:] If there exists $v' \in W^{\modelPML_2}$ such that $vR^{\modelPML_2}v'$ then there exists $w' \in W^{\modelPML_1}$ such that $wR^{\modelPML_1}w'$ and $(\modelPML_1, w')Z(\modelPML_2, v')$.
			
			\item[Zig$_\blacklozenge$:] If there exists $(\modelPML_1', w_1')$ such that $(\modelPML_1, w_1) \poisonRelation (\modelPML_1', w_1')$, then there exists $(\modelPML_2', w_2')$ such that $(\modelPML_2, w_2) \poisonRelation (\modelPML_2', w_2')$ and $(\modelPML_1', w_1')Z(\modelPML_2', w_2')$.
			
			\item[Zag$_\blacklozenge$:] If there exists $(\modelPML_2', w_2')$ such that $(\modelPML_2, w_2) \poisonRelation (\modelPML_2', w_2')$, then there exists $(\modelPML_1', w_1')$ such that $(\modelPML_1, w_1) \poisonRelation (\modelPML_1', w_1')$ and $(\modelPML_1', w_1')Z(\modelPML_2', w_2')$.
		\end{description}
	\end{definition}
	
	An example of a p-bisimulation relation is depicted in Figure \ref{fig:exBisi2}.
	Observe that, unlike for bisimulation in the standard modal language, p-bisimulation involves transitions between pointed models in the clauses Zig$_\blacklozenge$ and Zag$_\blacklozenge$. For a simple example of two models which are not p-bisimilar consider a model consisting of just one reflexive point, and its unraveling in an infinite chain. 

	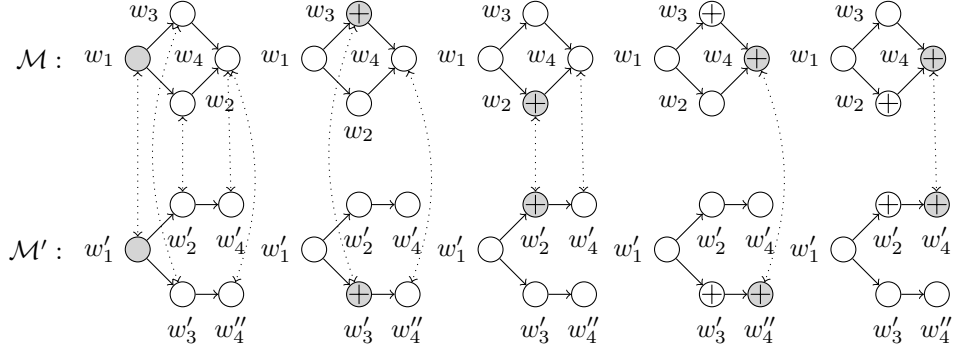
\begin{figure*}[t]
		\centering
		\begin{tikzpicture}[node distance = .3cm]
		\node[draw, circle, label = {[label distance= .005cm]180:$w_3$}] (M0) {};
		\node[draw = none, below = of M0] (M1) {};
		\node[draw, circle, fill={rgb:black,1;white,5}, left = of M1, label = {[label distance= .005cm]180:$w_1$}] (M2) {};
		\node[draw, circle, right = of M1, label = {[label distance= .005cm]180:$w_4$}] (M3) {};
		\node[draw, circle, below = of M1, label = {[label distance= .005cm]0:$w_2$}] (M4) {};
		\node[draw = none, left = .7cm of M2] (M) {$\modelPML:$};
		
		\path[->] (M2) edge [] node {} (M0);
		\path[->] (M2) edge [] node {} (M4);
		\path[->] (M4) edge [] node {} (M3);
		\path[->] (M0) edge [] node {} (M3);

		\node[draw, circle, below = 1cm of M4, label = {[label distance= .005cm]270:$w_2'$}] (N01) {};
		\node[draw, circle, right = of N01, label = {[label distance= .005cm]270:$w_4'$}] (N02) {};
		\node[draw = none, below = of N01] (N0Tmp) {};
		\node[draw, circle, fill={rgb:black,1;white,5}, left = of N0Tmp, label = {[label distance= .005cm]180:$w_1'$}] (N00) {};
		\node[draw, circle, below = of N0Tmp, label = {[label distance= .005cm]270:$w_3'$}] (N03) {};
		\node[draw, circle, right = of N03, label = {[label distance= .005cm]270:$w_4''$}] (N04) {};
		\node[draw = none, left = .7cm of N00] {$\modelPML':$}; 
		
		\path[->] (N00) edge [] node {} (N01);
		\path[->] (N00) edge [] node {} (N03);
		\path[->] (N01) edge [] node {} (N02);
		\path[->] (N03) edge [] node {} (N04);
		
		\path[<->] (M2) edge [dotted] node {} (N00);
		\path[<->] (M3) edge [bend left = 20, dotted] node {} (N04);
		\path[<->] (M4) edge [dotted] node {} (N01);
		\path[<->] (M3) edge [dotted] node {} (N02);
		\path[<->] (M0) edge [bend right = 20, dotted] node {} (N03);

		\node[draw, circle, plus, fill={rgb:black,1;white,5}, label = {[label distance= .005cm]180:$w_3$}, right = 2cm of M0] (M10) {};
		\node[draw = none, below = of M10] (M11) {};
		\node[draw, circle, left = of M11, label = {[label distance= .005cm]180:$w_1$}] (M12) {};
		\node[draw, circle, right = of M11, label = {[label distance= .005cm]180:$w_4$}] (M13) {};
		\node[draw, circle, below = of M11, label = {[label distance= .005cm]270:$w_2$}] (M14) {};
		
		\path[->] (M12) edge [] node {} (M14);
		\path[->] (M12) edge [] node {} (M10);
		\path[->] (M14) edge [] node {} (M13);
		\path[->] (M10) edge [] node {} (M13);

		\node[draw, circle, below = 1cm of M14, label = {[label distance= .005cm]270:$w_2'$}] (N11) {};
		\node[draw, circle, right = of N11, label = {[label distance= .005cm]270:$w_4'$}] (N12) {};
		\node[draw = none, below = of N11] (N1Tmp) {};
		\node[draw, circle, left = of N1Tmp, label = {[label distance= .005cm]180:$w_1'$}] (N10) {};
		\node[draw, circle, plus, fill={rgb:black,1;white,5}, below = of N1Tmp, label = {[label distance= .005cm]270:$w_3'$}] (N13) {};
		\node[draw, circle, right = of N13, label = {[label distance= .005cm]270:$w_4''$}] (N14) {};
		
		\path[->] (N10) edge [] node {} (N11);
		\path[->] (N10) edge [] node {} (N13);
		\path[->] (N11) edge [] node {} (N12);
		\path[->] (N13) edge [] node {} (N14);
		
		\path[<->] (M10) edge [dotted, bend right = 20] node {} (N13);
		\path[<->] (M13) edge [dotted, bend left = 20] node {} (N14);

		\node[draw, circle, right = 2cm of M10, label = {[label distance= .005cm]180:$w_3$}] (M20) {};
		\node[draw = none, below = of M20] (M21) {};
		\node[draw, circle, left = of M21, label = {[label distance= .005cm]180:$w_1$}] (M22) {};
		\node[draw, circle, right = of M21, label = {[label distance= .005cm]180:$w_4$}] (M23) {};
		\node[draw, circle, plus, fill={rgb:black,1;white,5}, below = of M21, label = {[label distance= .005cm]180:$w_2$}] (M24) {};
		
		\path[->] (M22) edge [] node {} (M24);
		\path[->] (M22) edge [] node {} (M20);
		\path[->] (M24) edge [] node {} (M23);
		\path[->] (M20) edge [] node {} (M23);

		\node[draw, circle, plus, fill={rgb:black,1;white,5}, below = 1cm of M24, label = {[label distance= .005cm]270:$w_2'$}] (N21) {};
		\node[draw, circle, right = of N21, label = {[label distance= .005cm]270:$w_4'$}] (N22) {};
		\node[draw = none, below = of N21] (N2Tmp) {};
		\node[draw, circle, left = of N2Tmp, label = {[label distance= .005cm]180:$w_1'$}] (N20) {};
		\node[draw, circle, below = of N2Tmp, label = {[label distance= .005cm]270:$w_3'$}] (N23) {};
		\node[draw, circle, right = of N23, label = {[label distance= .005cm]270:$w_4''$}] (N24) {};
		
		\path[->] (N20) edge [] node {} (N21);
		\path[->] (N20) edge [] node {} (N23);
		\path[->] (N21) edge [] node {} (N22);
		\path[->] (N23) edge [] node {} (N24);
		
		\path[<->] (M24) edge [dotted] node {} (N21);
		\path[<->] (M23) edge [dotted] node {} (N22);

		\node[draw, circle, plus, right = 2cm of M20, label = {[label distance= .005cm]180:$w_3$}] (M30) {};
		\node[draw = none, below = of M30] (M31) {};
		\node[draw, circle, left = of M31, label = {[label distance= .005cm]180:$w_1$}] (M32) {};
		\node[draw, circle, plus, fill={rgb:black,1;white,5}, right = of M31, label = {[label distance= .005cm]180:$w_4$}] (M33) {};
		\node[draw, circle, below = of M31, label = {[label distance= .005cm]180:$w_2$}] (M34) {};
		
		\path[->] (M32) edge [] node {} (M34);
		\path[->] (M32) edge [] node {} (M30);
		\path[->] (M34) edge [] node {} (M33);
		\path[->] (M30) edge [] node {} (M33);

		\node[draw, circle, below = 1cm of M34, label = {[label distance= .005cm]270:$w_2'$}] (N31) {};
		\node[draw, circle, right = of N31, label = {[label distance= .005cm]270:$w_4'$}] (N32) {};
		\node[draw = none, below = of N31] (N3Tmp) {};
		\node[draw, circle, left = of N3Tmp, label = {[label distance= .005cm]180:$w_1'$}] (N30) {};
		\node[draw, circle, plus, below = of N3Tmp, label = {[label distance= .005cm]270:$w_3'$}] (N33) {};
		\node[draw, circle, plus, fill={rgb:black,1;white,5}, right = of N33, label = {[label distance= .005cm]270:$w_4''$}] (N34) {};
		
		\path[->] (N30) edge [] node {} (N31);
		\path[->] (N30) edge [] node {} (N33);
		\path[->] (N31) edge [] node {} (N32);
		\path[->] (N33) edge [] node {} (N34);
		
		\path[<->] (M33) edge [dotted, bend left = 20] node {} (N34);

		\node[draw, circle, right = 2cm of M30, label = {[label distance= .005cm]180:$w_3$}] (M40) {};
		\node[draw = none, below = of M40] (M41) {};
		\node[draw, circle, left = of M41, label = {[label distance= .005cm]180:$w_1$}] (M42) {};
		\node[draw, circle, plus, fill={rgb:black,1;white,5}, right = of M41, label = {[label distance= .005cm]180:$w_4$}] (M43) {};
		\node[draw, circle, plus, below = of M41, label = {[label distance= .005cm]180:$w_2$}] (M44) {};
		
		\path[->] (M42) edge [] node {} (M44);
		\path[->] (M42) edge [] node {} (M40);
		\path[->] (M44) edge [] node {} (M43);
		\path[->] (M40) edge [] node {} (M43);

		\node[draw, circle, plus, below = 1cm of M44, label = {[label distance= .005cm]270:$w_2'$}] (N41) {};
		\node[draw, circle, plus, fill={rgb:black,1;white,5}, right = of N41, label = {[label distance= .005cm]270:$w_4'$}] (N42) {};
		\node[draw = none, below = of N41] (N4Tmp) {};
		\node[draw, circle, left = of N4Tmp, label = {[label distance= .005cm]180:$w_1'$}] (N40) {};
		\node[draw, circle, below = of N4Tmp, label = {[label distance= .005cm]270:$w_3'$}] (N43) {};
		\node[draw, circle, right = of N43, label = {[label distance= .005cm]270:$w_4''$}] (N44) {};
		
		\path[->] (N40) edge [] node {} (N41);
		\path[->] (N40) edge [] node {} (N43);
		\path[->] (N41) edge [] node {} (N42);
		\path[->] (N43) edge [] node {} (N44);
		
		\path[<->] (M43) edge [dotted] node {} (N42);
		\end{tikzpicture}
		\caption{Two p-bisimilar models $\Model$ and $\Model'$ (leftmost models), and the models reachable via poisoning from them. The shadow world is the pointed one and crosses denote poisoned states. Dotted lines represent links of the p-bisimulation instantiating conditions of Definition \ref{def:pbisimulation}.}
		\label{fig:exBisi2}
	\end{figure*}

	\begin{remark}[Argumentation and p-bisimulation]
		As argued in \cite{gabbay14when}, modal bisimulation formalizes a natural notion of similarity of argumentation frameworks that preserves important argumentation-theoretic notions.
		It has for instance been shown \cite[Th. 6]{grossi10logic} that, given two totally bisimilar models $\Model_1$ and $\Model_2$, a set of arguments denoted by $p$ in $\Model_1$ is admissible (respectively, complete, stable or grounded) in the frame of $\Model_1$, if and only if the set of arguments denoted by $p$ in $\Model_2$ is admissible (respectively, complete, stable or grounded) in the frame of $\Model_2$. Which strengthening of the notion of bisimulation is needed to guarantee the preservation of credulous admissibility across frameworks was mentioned as an open question in \cite{gabbay14when}. P-bisimulation provides an elegant answer.
	\end{remark}

	\subsection{Characterization}
	
	The aim of this section is to establish a characterization theorem (Theorem \ref{th:characterization}) in the tradition of \cite{VanBenthem_1983}. The standard proof methods can be adapted easily to fit \PML.
	We start by precisely relating p-bisimulation with poison modal equivalence.
	
	\begin{theorem}
		\label{thm:pbisi->equiv}
		For two pointed models $(\modelPML_1, w_1)$ and $(\modelPML_2, w_2)$, if $(\modelPML_1, w_1) \bisimulation (\modelPML_2, w_2)$ then $(\modelPML_1, w_1) \equivModel (\modelPML_2, w_2)$.
	\end{theorem}
	\begin{proof}[Sketch of proof]
		The proof is by induction on the structure of formulas.
		The base case is covered by the atomic condition of the definition of p-bisimulation. 
		For the inductive case, we provide details only for the $\blacklozenge$ modality. Let $Z$ be the p-bisimulation relation. Suppose that $(\modelPML_1, w_1) \models \blacklozenge \phi$, then given the semantics of $\blacklozenge$, there exists $(\modelPML_1', w_1')$ such that $(\modelPML_1, w_1) \poisonRelation (\modelPML_1', w_1')$ and $(\modelPML_1', w_1') \models \phi$. From the clause Zig$_\blacklozenge$ of a p-bisimulation we know that there exists $(\modelPML_2', w_2')$ such that $(\modelPML_2, w_2) \poisonRelation (\modelPML_2', w_2')$ and $(\modelPML_1, w_1') Z (\modelPML_2, w_2')$. By the induction hypothesis, we have $(\modelPML_1', w_1') \equivModel (\modelPML_2', w_2')$ which brings that $(\modelPML_2', w_2') \models \phi$, from which we conclude $(\modelPML_2, w_2) \models \blacklozenge \phi$. The direction from $(\modelPML_2, w_2) \models \blacklozenge \phi$ to $(\modelPML_1, w_1) \models \blacklozenge \phi$ is similar and uses the Zag$_\blacklozenge$ condition. 
	\end{proof}

	\begin{remark}[Credulous admissibility and p-bisimulation]
		Formula \eqref{eq:PRO} expresses the existence of credulous admissible arguments (Remark \ref{rem:credadm}), and is invariant for p-bisimulation (Theorem \ref{thm:pbisi->equiv}). It directly follows that, given two p-bisimilar pointed models $(\Model_1, w_1)$ and $(\Model_2, w_2)$, the frame of $\Model_1$ contains credulously admissible arguments if and only if the frame of $\Model_2$ does.
	\end{remark}
	
	For the converse result, some auxiliary definitions are needed.\footnote{Cf. \cite[Ch. 2]{Blackburn_2001}.} Let $\modelPML = (W, R, V)$ be a model. A set of \FOL{} formulas $\Gamma(x)$ from $\lanFOL$ with one free variable $x$ is realized by $\modelPML$ if there exists $w \in W$ s.t. $\modelPML \models \Gamma(x)[x:=w]$. We say that $\modelPML$ (viewed as a \FOL{} structure) is $\omega$-saturated if for any finite set $X \subseteq W$, the expansion $\modelPML_X$ realizes every set $\Gamma(x) \in \lanFOL_X$ (i.e., the expansion of $\lanFOL$ with constants for the elements in $X$) whenever every finite subset $\Gamma'(x) \subseteq \Gamma(x)$ is realized in $\modelPML_X$.

	\begin{theorem}
		\label{thm:equiv->pbisi}
		For any two $\omega$-saturated models $(\modelPML_1, w_1)$ and $(\modelPML_2, w_2)$, if $(\modelPML_1, w_1) \equivModel (\modelPML_2, w_2)$ then $(\modelPML_1, w_1) \bisimulation (\modelPML_2, w_2)$.
	\end{theorem}
	\begin{proof}[Sketch of proof]
		We show that $\equivModel$ is itself a p-bisimulation. The base case holds trivially. The proof for the Zig$_\lozenge$ and Zag$_\lozenge$ proceed in the usual manner. We need to prove that the conditions Zig$_\blacklozenge$ and Zag$_\blacklozenge$ are verified. \fbox{Zig$_\blacklozenge$} Let us assume that $(\modelPML_1, w_1) \equivModel (\modelPML_2, w_2)$ and that there exists a pointed model $(\modelPML_1', w_1')$ such that $(\modelPML_1, w_1) \poisonRelation (\modelPML_1', w_1')$. We show that there exists $(\modelPML_2', w_2')$ such that $(\modelPML_2, w_2) \poisonRelation (\modelPML_2', w_2')$ and $(\modelPML_1', w_1') \equivModel (\modelPML_2', w_2')$. First, observe that for any finite $\Gamma \subseteq \poisonTheory(\modelPML_1', w_1')$ (see Definition \ref{def:poisonModalTheory} for $\poisonTheory$), by Theorem \ref{thm:correctTranslFOL}, the following equivalences hold:
		\begin{align*}
		(\modelPML_1, w_1) \models & \blacklozenge \bigwedge \Gamma 
		\Leftrightarrow (\modelPML_2, w_2) \models \blacklozenge \bigwedge \Gamma \\
		\Leftrightarrow & \modelPML_2 \models ST_x^\emptyset\left(\blacklozenge \bigwedge \Gamma \right)[x := w_2] \\
		\Leftrightarrow & \modelPML_2 \models \exists y \left(xR^{\modelPML_2}y \wedge ST_y^{\{y\}}\left(\bigwedge \Gamma\right)\right)[x := w_2].
		\end{align*}
		Since $\modelPML_2$ is $\omega$-saturated by assumption, we have:
		\[
		\exists y \in \modelPML_2, \modelPML_2 \models ST_y^{\{y\}}\left(\mathbb{T}^p(\modelPML_1', w_1')\right).
		\]
		
		By Theorem \ref{thm:correctTranslFOL}, there exists a pointed model $(\modelPML_2', w_2')$ such that $(\modelPML_2, w_2) \poisonRelation (\modelPML_2', w_2')$ and 
		$
		\modelPML_2' \models ST_x^\emptyset \left(\poisonTheory (\modelPML_1', w_1') \right) [x := w_2'].
		$
		Again by Theorem \ref{thm:correctTranslFOL}, it follows that $(\modelPML_1', w_1') \equivModel (\modelPML_2', w_2')$. \fbox{Zag$_\blacklozenge$} The proof is by a similar argument.
	\end{proof}

	\begin{theorem}\label{th:characterization}
		An $\lanFOL$ formula is equivalent to the translation of an $\languagePML$ formula if and only if it is invariant for p-bisimulation.
	\end{theorem}
	\begin{proof}
		\fbox{Left-to-right} By Theorem \ref{thm:pbisi->equiv}. \fbox{Right-to-left} Let $\phi \in \lanFOL$ be a formula with only one free variable $x$. Let us assume that $\phi$ is invariant under p-bisimulation and let $\mathcal{C}$ be the following set: $\mathcal{C}(\phi) = \{ST_x^\emptyset(\psi)\mid\psi \in \languagePML\text{ and } \phi \models ST_x^\emptyset(\psi)\}$.
		Let us now show that $\mathcal{C}(\phi) \models \phi$, that is to say: for any pointed model $(\modelPML, w)$ if $\modelPML \models \mathcal{C}(\phi)[x:=w]$ then $\modelPML \models \phi[x:=w]$. 
		To do so, let us first show that $\Sigma = ST_x^\emptyset\left(\poisonTheory(\modelPML, w)\right)\cup\{\phi\}$ is consistent. 
		
		Let us assume for the sake of the contradiction that $\Sigma$ is inconsistent. By the compactness of \FOL{}, $\models \phi \rightarrow \neg \bigwedge \Gamma$ for some finite $\Gamma \subseteq ST_x^\emptyset\left(\poisonTheory(\modelPML, w)\right)$. By the definition of $\mathcal{C}(\phi)$ this means that $\neg \bigwedge \Gamma \in \mathcal{C}(\phi)$ and so $\neg \bigwedge \Gamma \in ST_x^\emptyset\left(\poisonTheory(\modelPML, w)\right)$ which is in contradiction with $\Gamma \subseteq ST_x^\emptyset \left(\poisonTheory(\modelPML, w)\right)$.
		
		Let us now show that $\modelPML \models \phi[x:=w]$. As $\Sigma$ is consistent, there exists a pointed model $(\modelPML', w')$ such that $(\modelPML', w') \models \Sigma$. From this it is immediate to see that $(\modelPML, w) \equivModel (\modelPML', w')$. Let us now consider two $\omega$-saturated elementary extensions $(\modelPML_\omega, w)$ and $(\modelPML_\omega', w')$ of $(\modelPML, w)$ and $(\modelPML', w')$. Such extensions exist by standard argument (cf. \cite[Proposition 3.2.6]{chang73model}). As \FOL{} is invariant under elementary extensions, from $\modelPML'\models \phi[x:=w]$ we can conclude that $\modelPML_\omega'\models \phi[x:=w]$. As we have assumed that $\phi$ is invariant under p-bisimulation and thanks to Theorem \ref{thm:equiv->pbisi}, we get $\modelPML_\omega \models \phi[x:=w]$ which brings $\modelPML \models \phi[x:=w]$. We have thus shown that $\mathcal{C}(\phi) \models \phi$.
		
		To conclude the proof, we need to show that $\mathcal{C}(\phi) \models \phi$ implies that $\phi$ is equivalent to the translation of an $\languagePML$ formula. As $\mathcal{C}(\phi) \models \phi$, from the deduction and the compactness theorems of \FOL, $\models \bigwedge \Gamma \rightarrow \phi$ for some finite $\Gamma \subset \mathcal{C}(\phi)$. By definition of $\mathcal{C}(\phi)$ we also have $\models \phi \rightarrow \bigwedge \Gamma$ and so $\models \phi \leftrightarrow \bigwedge \Gamma$.
	\end{proof}
	Finally, it is worth mentioning a simple example of a property which is not expressible in \PML. Properties of this type are for instance those involving counting quantifiers, e.g.: the current state has at least $n$ successors. It is easy to devise two p-bisimilar models for which the above property holds for one, but not the other. %

	\section{Undecidability} \label{sec:Decidability}
	
	In this section we tackle the question of the decidability of \PML.

	\subsection{Undecidability of $\PML_3$}
	
	In this section we establish the undecidability of $\PML_3$, that is the variant of $\PML$ with three standard modalities and three poison modalities, each with a distinct poison atom that correspond to language $\languagePML_3$. We call $R, R_1$ and $R_2$ the three accessibility relations of a model of $\PML_3$. 
	The satisfaction problem for $\PML_3$ can be defined as follows:
	\begin{description}
		\item[Data:] A $\PML_3$ formula $\phi \in \languagePML_3$.
		\item[Problem:] Is there $(\modelPML, w)$, with $\modelPML \in \mathfrak{M}^\emptyset$, such that $(\modelPML, w) \models \phi$ ?
	\end{description}
	
	\begin{theorem}
		\label{thm:undecidability}
		The satisfaction problem for $\PML_3$ is undecidable.
	\end{theorem}
	\begin{proof}[Sketch of proof]
		We reduce the problem of the $\mathbb{N} \times \mathbb{N}$ tilling in a similar way as for the proof of undecidability of hybrid logic $\mathcal{H}(\downarrow)$ presented in \cite{Ten_2005}.
		Let us recall $\mathbb{N} \times \mathbb{N}$ tilling problem. Given a finite set of colors $C$, a tile is a 4-tuple of colors (its 4 sides). The $\mathbb{N} \times \mathbb{N}$ tilling problem is then defined as follows:
		\begin{description}
			\item[Data:] A finite set $T$ of tiles. 
			\item[Problem:] Can the infinite grid $\mathbb{N} \times \mathbb{N}$ be tiled using only tiles in $T$ and such that two adjacent tiles share the same color on their common edge ?
		\end{description}
		This problem is known to be undecidable \cite{Harel_1983}. So let $T$ be a finite set of tiles. We claim that the following formula is satisfiable if and only if $T$ tiles the grid $\mathbb{N} \times \mathbb{N}$:
		\begin{align}
		\phi_T = \alpha \wedge \beta \wedge \gamma \wedge \square\left(\delta^1_T \wedge \delta^2_T \wedge \delta^3_T\right)
		\end{align} 
		with $\alpha, \beta, \gamma, \delta^1_T, \delta^2_T$ and $\delta^3_T$ such as:
		\begin{align*}
		\alpha & = q \wedge \square (\neg q \wedge \lozenge q) \wedge \square \blacksquare_1 \lozenge (q \wedge \lozenge \poisonAtom) \wedge \square \blacksquare_2 \lozenge(q \wedge \lozenge \poisonAtom) \\
		\beta & = \bigwedge_{i = 1, 2} \left(\square \lozenge_i \top \wedge \blacksquare \square (q \rightarrow \square (\lozenge_i \poisonAtom \rightarrow \square_i \poisonAtom))\right) \\
		\gamma & = \blacksquare \square \left(q \rightarrow \square (\square_1 \square_2 \neg \poisonAtom \vee \square_2 \square_1 \poisonAtom)\right) \\
		\delta_T^1 & = \bigvee_{t \in T} \left(p_t \wedge \bigwedge_{t'\in T, t' \neq t}\neg p_{t'}\right) \\
		\delta_T^2 & = \bigwedge_{t \in T}\left(p_t \rightarrow \square_2 \bigvee_{t' \in T, \mathit{left}(t') = \mathit{right}(t)} p_t\right) \\
		\delta_T^3 & = \bigwedge_{t \in T}\left(p_t \rightarrow \square_1 \bigvee_{t' \in T, \mathit{bottom}(t') = \mathit{top}(t)} p_t\right).
		\end{align*}
		
		In these formulas, the modality $\lozenge_1$ represent vertical moves on the grid, the modality $\lozenge_2$ horizontal moves and the modality $\lozenge$ moves from any point to any point on the grid (i.e., a universal modality). For a tile $t\in T$, the predicate $p_t$ models the fact that $t$ is placed on the point and the four predicates $\mathit{top}(t), \mathit{right}(t), \mathit{bottom}(t)$ and $\mathit{left}(t)$ represent the four colors of $t$. For $(\modelPML, w) \models \phi_T$, the sub-formulas of $\phi_T$ are interpreted as follow:
		\begin{itemize}
			\item $\alpha$: $w$ is a $q$-world, its $R$-successors are not $q$ and link back to it, and the set of its $R$-successors is closed under $R_1$ and $R_2$.
			\item $\beta$: for all $R$-successor of $w$, accessibility relations $R_1$ and $R_2$ are total functions. 
			\item $\gamma$: accessibility relations $R_1$ and $R_2$ commute.
			\item $\square(\delta_T^1 \wedge \delta_T^2 \wedge \delta_T^3)$: only one tile is present at each node and horizontal and vertical tiling are correct.
		\end{itemize}
		By the construction of $\phi_T$, it is easy to see that $\phi_T$ is satisfiable if and only if there is a tiling of the grid $\mathbb{N} \times \mathbb{N}$ with $T$ (Figure \ref{fig:gridUndecidability}).
	\end{proof}
	
	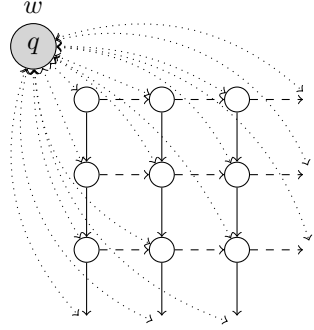
\begin{figure}
		\centering
		\begin{tikzpicture}
		\node[draw, circle] (01) {};
		\node[draw, circle, left of=01] (00) {};
		\node[draw, circle, right of=01] (02) {};
		\node[right of=02] (03) {};
		
		\node[draw, circle, below of=01] (11) {};
		\node[draw, circle, left of=11] (10) {};
		\node[draw, circle, right of=11] (12) {};
		\node[right of=12] (13) {};
		
		\node[draw, circle, below of=11] (21) {};
		\node[draw, circle, left of=21] (20) {};
		\node[draw, circle, right of=21] (22) {};
		\node[right of=22] (23) {};
		
		\node[below of=21] (31) {};
		\node[left of=31] (30) {};
		\node[right of=31] (32) {};
		
		\node[draw, circle, fill={rgb:black,1;white,5}, label = {[label distance= .005cm]90:$w$}, above left of = 00] (w) {$q$};
		
		\path[<->] (w) edge [dotted] node {} (00);
		\path[<->] (w) edge [dotted] node {} (01);
		\path[<->] (w) edge [dotted, bend left = 15] node {} (02);
		\path[<->] (w) edge [dotted, bend left] node {} (03);
		\path[<->] (w) edge [dotted] node {} (10);
		\path[<->] (w) edge [dotted, bend left = 15] node {} (11);
		\path[<->] (w) edge [dotted, bend left] node {} (12);
		\path[<->] (w) edge [dotted, bend left = 40] node {} (13);
		\path[<->] (w) edge [dotted, bend right = 15] node {} (20);
		\path[<->] (w) edge [dotted, bend right] node {} (21);
		\path[<->] (w) edge [dotted, bend left = 17] node {} (22);
		\path[<->] (w) edge [dotted, bend left = 35] node {} (23);
		\path[<->] (w) edge [dotted, bend right] node {} (30);
		\path[<->] (w) edge [dotted, bend right = 40] node {} (31);
		\path[<->] (w) edge [dotted, bend right = 33] node {} (32);
		
		\path[->] (00) edge [dashed] node {} (01);
		\path[->] (00) edge [] node {} (10);
		\path[->] (01) edge [dashed] node {} (02);
		\path[->] (01) edge [] node {} (11);
		\path[->] (02) edge [dashed] node {} (03);
		\path[->] (02) edge [] node {} (12);
		
		\path[->] (10) edge [dashed] node {} (11);
		\path[->] (10) edge [] node {} (20);
		\path[->] (11) edge [dashed] node {} (12);
		\path[->] (11) edge [] node {} (21);
		\path[->] (12) edge [dashed] node {} (13);
		\path[->] (12) edge [] node {} (22);
		
		\path[->] (20) edge [dashed] node {} (21);
		\path[->] (20) edge [] node {} (30);
		\path[->] (21) edge [dashed] node {} (22);
		\path[->] (21) edge [] node {} (31);
		\path[->] (22) edge [dashed] node {} (23);
		\path[->] (22) edge [] node {} (32);
		\end{tikzpicture}
		\caption{A  model of formula $\phi_T$. $R$ is represented by dotted links, $R_1$ by dashed link and $R_2$ by plain links, the shadow world is the pointed one.}
		\label{fig:gridUndecidability}
	\end{figure}

	\subsection{Failure of FMP for \PML}
	
	It is unclear however whether \PML{} {\em with only} one standard and one poison modality is also undecidable. We suspect it is, and we can show it fails to have the finite model property.
	
	\begin{proposition} \label{prop:infinite}
		\PML{} does not have the finite model property.
	\end{proposition}
	\begin{proof}
		We provide a formula whose models are all infinite. Let us consider $\phi = \alpha \land \beta \land \gamma \land \delta \land \epsilon$ with the sub-formulas defined below.
		\begin{itemize}
			\item $\alpha = \neg q \land \lozenge \top \land \square q \land \square (\lozenge \top \land \square \neg q)$: the current state falsifies $q$ and all its successors (there exists at least one) are $q$ and have in turn successors (at least one) which all falsify $q$.
			\item $\beta = \blacksquare \square \lozenge \poisonAtom$: after any poisoning a state is reached whose successors can reach the poisoned state in one step. In other words,
			all successors of the current state have successors linked via symmetric edges.
			\item $\gamma = \blacksquare \square \lozenge (\neg q \land \lozenge \poisonAtom) \land \square \square \neg \blacklozenge \lozenge \poisonAtom$: after any poisoning a state is reached whose successors are not reflexive loops (right conjunct), and can reach a $\neg q$ state which can in turn reach the poison state. In other words, all successors of the current state lay on cycles of length $3$.
			\item $\delta = \square \square \blacksquare \square (q \rightarrow \lozenge \poisonAtom)$: all successors of the current state's successors are such that after any poisoning, and  further $q$-successor can reach back to the poisoned state.
			\item $\epsilon = \square \blacklozenge \neg \lozenge (q \land \lozenge( \neg q \land \lozenge \poisonAtom)))$: all successors of the current state are such that there is one successor that can be poisoned and such that none of its successors satisfies $q$ and can reach the poisoned state in two steps via a $\neg q$ state.  
		\end{itemize}
		Now let $(\modelPML, w) \models \phi$. Then, $w$ is followed by distinct successors $w'$ ($\alpha$) that have successors $w''$ which are linked back to their predecessors $w'$ by symmetric edges ($\beta$). These $w''$ states also have successors, different from $w'$ which also have $w'$ as successor ($\gamma$) and which are also successors of $w'$ ($\delta$). Hence, $w''$ is followed by an infinite path of distinct states. Finally, there exists one such $w''$ which has no other predecessor than $w'$ ($\epsilon$), that is, $w''$ is the root of an infinite sequence of distinct states which are all successors of $w'$. One such model is depicted in Figure \ref{fig:infinite}.
	\end{proof}

	\begin{figure}
		\centering
		\begin{tikzpicture}[node distance = .5cm]
		\node[draw, circle, fill={rgb:black,1;white,5}, label = {[label distance= .01cm] 180:$w$}] (0) {};
		\node[draw, circle, label = {[label distance = .01cm]-15:$w'$}, right = of 0] (1) {$q$};
		\node[draw, circle, label = {[label distance = .01cm]90:$w_1''$}, above = of 1] (2) {};
		\node[draw, circle, label = {[label distance = .01cm]90:$w_2''$}, right = of 2] (3) {};
		\node[draw, circle, label = {[label distance = .01cm]90:$w_3''$}, right = of 3] (4) {};
		\node[draw = none, right =  of 4] (5) {\ldots};
		
		\path[->] (0) edge [] node {} (1);
		\path[<->] (1) edge [] node {} (2);
		\path[<->] (1) edge [] node {} (3);
		\path[<->] (1) edge [] node {} (4);
		\path[<->] (1) edge [] node {} (5);
		\path[->] (2) edge [] node {} (3);
		\path[->] (3) edge [] node {} (4);
		\path[->] (4) edge [] node {} (5);
		\end{tikzpicture}
		\caption{An infinite model of $\phi$ from Proposition \ref{prop:infinite}, the shadow world is the pointed one.}
		\label{fig:infinite}
	\end{figure}
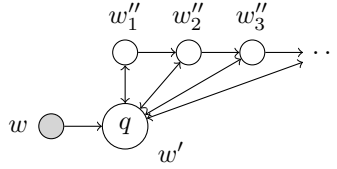

	\section{Discussion \& Conclusions} \label{sec:conclusions}

	In this last section we relate \PML{} with memory and hybrid logics establishing results about their relative expressivity (Propositions \ref{prop:memory} and \ref{prop:hybrid}). Such results are useful to position \PML{} precisely within the landscape of existing extensions of the standard modal language. We then conclude by charting a few lines of future research.
	
	\subsection{\PML{} and Memory Logics}
	
	As hinted at in the introduction, \PML{} is tightly related to memory logics \cite{areces08expressive,areces11expressive}. The simplest memory logic, $\Model(\textcircled{r},\textcircled{k})$, extends modal semantics by considering frames $\tuple{W,R,M}$ where $M \subseteq W$ is a set of states that have been `memorized'. 
	The standard modal language is then extended with two operators $\textcircled{r}$ and $\textcircled{k}$ defined as: 
	\begin{align*}
	(\tuple{W,R,M, V}, w) & \models \textcircled{r} \phi \Longleftrightarrow (\tuple{W,R,M \cup \{ w \}, V}, w) \models \phi \\
	(\tuple{W,R,M, V}, w) & \models \textcircled{k} \Longleftrightarrow w \in M,
	\end{align*}
	where $V$ is a valuation function. Intuitively, the $\textcircled{r}$ stores the current state in the memory $M$, serving a similar purpose to our poisoning operation, and the nullary operator $\textcircled{k}$ works precisely as our atom $\poisonAtom$. Intuitively, \PML{} can be seen as a memory logic in which storing states via $\textcircled{r}$ occurs only after traversing an edge in the underlying frame.
	Technically, \PML{} is a proper fragment of $\Model(\textcircled{r},\textcircled{k})$:
	
	\begin{proposition} \label{prop:memory}
		$\Model(\textcircled{r},\textcircled{k})$ is strictly more expressive than  \PML.
	\end{proposition}
	\begin{proof}
		First of all observe that \PML{} models and $\Model(\textcircled{r},\textcircled{k})$ are exactly the same type of structures, where $M$ of the latter type of models corresponds to the truth-set of $\poisonAtom$ in the former type of models. We show that $\Model(\textcircled{r},\textcircled{k})$ is at least as expressive as \PML, by providing a truth-preserving embedding of the latter into the former. Such an embedding is provided by the following translation (clauses for Boolean connectives and $\lozenge$ are omitted as straightfoward):
		\begin{align*}
		MT(\poisonAtom) & = \textcircled{k} \\
		MT(\blacklozenge \phi) & = \lozenge \textcircled{r} MT(\phi)
		\end{align*}
		It is easy to see that such translation is truth-preserving. 
		
		To show that $\PML$ is strictly less expressive than $\Model(\textcircled{r},\textcircled{k})$ it suffices to provide two p-bisimilar pointed models which can be distinguished by a formula of $\Model(\textcircled{r},\textcircled{k})$. Such models are depicted in Figure \ref{fig:memory}. The model on the right satisfies formula $\textcircled{r} \lozenge \lozenge \textcircled{k}$ while the model on the left falsifies it.
	\end{proof}
	
	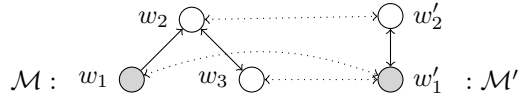
\begin{figure}[t]
		\centering
		\begin{tikzpicture}[node distance = .5cm]
		\node[draw, circle, label = {[label distance= .005cm]180:$w_2$}] (M0) {};
		\node[draw = none, below = of M0] (M1) {};
		\node[draw, circle, fill={rgb:black,1;white,5}, left = of M1, label = {[label distance= .005cm]180:$w_1$}] (M2) {};
		\node[draw, circle, right = of M1, label = {[label distance= .005cm]180:$w_3$}] (M3) {};
		\node[draw = none, left = .7cm of M2] (M) {$\modelPML:$};
		
		\path[->] (M2) edge [] node {} (M0);
		\path[<->] (M3) edge [] node {} (M0);
		
		\node[draw, circle, fill={rgb:black,1;white,5}, right = 1.5cm of M3, label = {[label distance= .005cm]0:$w_1'$}] (N0) {};
		\node[draw, circle, above = of N0, label = {[label distance= .005cm]0:$w_2'$}] (N4) {};
		\node[draw = none, right = .7cm of N0] (N) {$:\modelPML'$};
		
		\path[<->] (N0) edge [] node {} (N4);
		\path[<->] (M3) edge [dotted] node {} (N0);
		\path[<->] (M0) edge [dotted] node {} (N4);
		\path[<->] (M2) edge [bend left = 20, dotted] node {} (N0);
		\end{tikzpicture}
		\caption{Two p-bisimilar models from Proposition \ref{prop:memory}, the shadow world is the pointed one.}
		\label{fig:memory}
	\end{figure}

	\subsection{\PML{} and Hybrid Logics}
	
	\PML{} has, perhaps unsurprisingly, also tight links with hybrid logics. We show how \PML{} can be embedded into $\mathcal{H}(\binder)$ defined by \cite{Ten_2005}:
	$$\languageHL : \phi := p \mid i \mid \neg \phi \mid \phi \wedge \phi \mid \lozenge \phi \mid \binder x. \phi,$$
	with $p \in \propSet \cup \{\poisonAtom\}$ a propositional atom, and $i \in \nomSet$ a nominal. We write $\HLSatisfaction$ the satisfaction relation for $\mathcal{H}(\binder)$. Given an assignment $g : \nomSet \rightarrow W$, $g_m^x$ is called a $x$-variant of $g$ if $\forall i \in \nomSet, g(i) = g_m^x(i)$ and $g_m^x(x) = m$. The semantics is then $(M, g, m) \HLSatisfaction i \Leftrightarrow m = g(i)$ and $(M, g, m) \HLSatisfaction \binder x . \phi \Leftrightarrow (M, g_m^x, m) \HLSatisfaction \phi$. We can then set up the translation $HT^S:\languagePML \rightarrow \languageHL$ as follows, with $S \subseteq \nomSet$:
	\begin{align*}
	& HT^S(p) = p \\
	& HT^S(\poisonAtom) = \poisonAtom \vee \bigvee_{i\in N} i \\
	& HT^S(\neg \phi) = \neg HT^S (\phi) \\
	& HT^S(\phi \wedge \psi) = HT^S(\phi) \wedge HT^S(\psi) \\
	& HT^S(\lozenge \phi) = \lozenge HT^S(\phi) \\
	& HT^S(\blacklozenge \phi) = \lozenge \left(\binder x.HT ^{S \cup \{x\}} (\phi)\right),
	\end{align*}
	with $x$ a ''fresh variable'' never used before. 
	We also need a way to transform \PML{} models into hybrid models. Let $\modelPML = (W, R, V)$ be a \PML{}-model, we define $M = (W, R, V')$, the hybrid extension of $\modelPML$, by extending the valuation $V$ so that $\forall p \in \propSet \cup \{\poisonAtom\}, V'(p) = V(p)$, $\forall w\in W, \exists i \in \nomSet, V'(i) = \{w\}$ and $\forall i \in \nomSet, |V'(i)| = 1$. We can now proceed to show the translation $HT^S$ is correct.
	
	\begin{lemma}
		\label{lemma:HybridTransLemma}
		Let $\modelPML = (W, R, V)$ be a \PML{}-model and $M = (W, R, V')$ its hybrid extension. Let us consider $v, w \in W$ and $g$ an assignment. Then for $\phi \in \languagePML$ a \PML{}-formula and any set $S$, we have:
		$$(M^\poison_v, g, w) \HLSatisfaction HT^S(\phi) \Leftrightarrow (M, g_v^x, w) \HLSatisfaction HT^{S\cup\{x\}}(\phi).$$
	\end{lemma}
	\begin{proof}[Sketch of proof]
		We show this result by induction on the structure of $\phi$. The proof is trivial for the propositional case and the Boolean connectives, hence we only present cases for the poison atom and the poison modality.
		
		\fbox{$\phi = \poisonAtom$} The claim is proven by the following series of equivalences using the definition of the poison operation $\poison$:
		\begin{align*}
		(M^\poison_v, g, w) \HLSatisfaction HT^S(\poisonAtom) & \Leftrightarrow (M^\poison_v, g, w) \HLSatisfaction \poisonAtom \vee \bigvee_{i\in S} i \\
		& \Leftrightarrow (M, g, w) \HLSatisfaction \poisonAtom \vee \bigvee_{i \in S \cup \{x\}} i \\
		& \Leftrightarrow (M, g_v^x, w) \HLSatisfaction HT^{S\cup \{x\}}(\poisonAtom).
		\end{align*}
		
		\fbox{$\phi = \blacklozenge \psi$} with $\psi \in \languagePML$ The claim is proven by the following series of equivalences using the definition of $\poison$, the semantics of 
		$\binder$ and the induction hypothesis :
		\begin{align*}
		& (M^\poison_v, g, w) \HLSatisfaction HT^S(\blacklozenge \psi) \\
		\Leftrightarrow~ & (M^\poison_v, g, w) \HLSatisfaction \lozenge \left(\binder y . HT^{S\cup \{y\}}(\psi)\right) \\
		\Leftrightarrow~ & \exists u \in W, wRu, (M^\poison_v, g, u) \HLSatisfaction \binder y. HT^{S\cup\{y\}}(\psi) \\
		\Leftrightarrow~ & \exists u \in W, wRu, (M^\poison_v, g_u^y, u) \HLSatisfaction HT^{S\cup\{y\}}(\psi) \\
		\Leftrightarrow~ & \exists u \in W, wRu, (M, (g_u^y)_v^x, u) \HLSatisfaction HT^{S \cup \{x, y\}}(\psi) \\
		\Leftrightarrow~ & \exists u \in W, wRu, (M, g_v^x, u) \HLSatisfaction \binder y . HT^{S \cup \{x, y\}}(\psi) \\
		\Leftrightarrow~ & (M, g_v^x, w) \HLSatisfaction \lozenge \left(\binder y. HT^{S\cup\{x, y\}}(\psi) \right) \\
		\Leftrightarrow~ & (M, g_v^x, w) \HLSatisfaction HT^{S \cup \{x, y\}}(\blacklozenge \psi).
		\end{align*}
		This completes the proof.
	\end{proof}

	\begin{proposition} \label{prop:hybrid}
		Let $\modelPML = (W, R, V)$ be a \PML{}-model, $M = (W, R, V')$ its hybrid extension, $g$ an assignment and $\phi \in \languagePML$ a \PML{}-formula, we have:
		$$(\modelPML, w) \models \phi \Longleftrightarrow (M, g, w) \HLSatisfaction HT^\emptyset(\phi).$$
	\end{proposition}
	\begin{proof}[Sketch of proof]
		The proof is done by induction on the structure of $\phi$. We only present the non classical cases.
		
		\fbox{$\phi = \poisonAtom$} The claim is proven by the following series of equivalences using the defintion of a hybrid extension:
		\begin{align*}
		(\modelPML, w) \models \poisonAtom & \Leftrightarrow w \in V(\poisonAtom) \\
		& \Leftrightarrow w \in V'(\poisonAtom) \\
		& \Leftrightarrow (M, g, w) \HLSatisfaction \poisonAtom \\
		& \Leftrightarrow (M, g, w) \HLSatisfaction HT^\emptyset(\poisonAtom).
		\end{align*}
		
		\fbox{$\phi = \blacklozenge \psi$} with $\psi \in \languagePML$ The claim is proven by the following series of equivalences the definition of $\poison$, the semantics of 
		$\binder$, the induction hypothesis and Lemma \ref{lemma:HybridTransLemma}:
		\begin{align*}
		(\modelPML, w) \models \blacklozenge \psi & \Leftrightarrow \exists v \in W, wRv, (\modelPML^\poison_v, v) \models \psi \\
		& \Leftrightarrow \exists v \in W, wRv, (M^\poison_v, g, v) \HLSatisfaction HT^\emptyset(\psi) \\
		& \Leftrightarrow \exists v \in W, wRv, (M, g_v^x, v) \HLSatisfaction HT^{\{x\}}(\psi) \\
		& \Leftrightarrow \exists v \in W, wRv, (M, g, v) \HLSatisfaction \binder x. HT^{\{x\}}(\psi) \\
		& \Leftrightarrow (M, g, w) \HLSatisfaction \lozenge(\binder x. HT^{\{x\}}(\psi)) \\
		& \Leftrightarrow (M, g, w) \HLSatisfaction HT^\emptyset(\blacklozenge \psi).
		\end{align*}
		This completes the proof.
	\end{proof}

	\subsection{Discussion}
	
	These translations from \PML to memory logics and hybrid logics although very technical, have a particular interest when it comes to proof systems. Indeed we are now able to use the axiomatization and proof systems of these logics \cite{areces11expressive, Blackburn_2006_handbook} as prof systems for \PML, albeit too powerful for what we require. 
	
	\subsection{Conclusions}
	
	The paper has introduced and studied a modal logic \PML{} that arises naturally from a game-theoretic approach to a central decision problem in argumentation theory: the existence of credulously admissible sets.
	Our results provide new links between abstract argumentation theory \cite{dung95acceptability}, games on graphs \cite{berge96combinatorial} and modal logic \cite{Blackburn_2001}. Many directions for future research present themselves, at several levels. From the logic point of view, several technical problems remain open concerning \PML: Can the logic be axiomatised via a Hilbert calculus, possibly using insights from the memory logic literature (e.g., \cite{areces09completeness})? Can \PML{} be embedded in a fixed-variable fragment of \FOL{} (thereby shedding light on the precise level of saturation that would suffice for Theorem \ref{thm:equiv->pbisi})? Is \PML{} (with one modal operator and one sabotage operator) decidable? Can the logic be extended in a natural way with a least-fixpoint operator, for instance to express formulas $\eqref{eq:OPP}$ and $\eqref{eq:PRO}$? From the argumentation theory point of view, a natural question is whether the poison game can be adapted to capture the property of existence of skeptically admissible arguments, that is, arguments that belong to {\em all} admissible sets in a framework.

	\printbibliography
	
\end{document}